\documentclass[pra,preprint,superscriptaddress,amssymb]{revtex4}
\usepackage{graphicx}
\usepackage{bm}
\usepackage{amsmath,amssymb}
\usepackage{color}
\usepackage{amsthm} 

\theoremstyle{plain}
  \newtheorem{theorem}{Theorem}

\theoremstyle{definition}
  \newtheorem{definition}{Definition}

\theoremstyle{remark}
  
\theoremstyle{plain}
  \newtheorem*{theorem*}{Theorem}
  \newtheorem*{lemma*}{Lemma}
  \newtheorem*{corollary*}{Corollary}
  \newtheorem*{proposition*}{Proposition}
  \newtheorem*{claim*}{Claim}

\begin{document}
\title{Quantum interpretations of AWPP and APP} 
\author{Tomoyuki Morimae}
\email{morimae@gunma-u.ac.jp}
\affiliation{ASRLD Unit, Gunma University, 1-5-1 Tenjin-cho Kiryu-shi
Gunma-ken, 376-0052, Japan}
\author{Harumichi Nishimura}
\email{hnishimura@math.cm.is.nagoya-u.ac.jp}
\affiliation{
Graduate School of Information Science, 
Nagoya University,
Chikusa-ku, Nagoya, Aichi, 464-8601 Japan}

\begin{abstract}
AWPP is a complexity class introduced by Fenner, Fortnow, Kurtz, and Li,
which is defined using GapP functions.
Although it is an important class as the best upperbound of BQP,
its definition seems to be somehow artificial, and therefore it would be better if
we have some ``physical interpretation" of AWPP.
Here we provide a quantum physical interpretation of AWPP:
we show that AWPP is equal to the class of problems
efficiently solved by a quantum computer
with the ability of postselecting an event whose
probability is close to an FP function.
This result is applied to
also obtain a quantum physical interpretation of APP.
In addition, we consider a ``classical physical analogue" of these results,
and show that a restricted version
of ${\rm BPP}_{\rm path}$ contains ${\rm UP}\cap{\rm coUP}$
and is contained in WAPP.
\end{abstract}
\maketitle

\section{Introduction}
AWPP is a complexity class introduced by Fenner, Fortnow, Kurtz, and Li~\cite{toolkit}
to understand the structure of counting complexity classes
(see also Refs.~\cite{Li,Fenner}). It is defined as follows:

\begin{definition}
\label{def:AWPP}
A language $L$ is in ${\rm AWPP}$ iff 
for any polynomial $r$, there exist $f\in {\rm FP}$ and $g\in {\rm GapP}$ such that
for all $w$, $f(w)>0$ and
\begin{itemize}
\item[1.]
If $w\in L$ then $1-2^{-r(|w|)}\le \frac{g(w)}{f(w)}\le 1$.
\item[2.]
If $w\notin L$ then $0\le \frac{g(w)}{f(w)}\le2^{-r(|w|)}$.
\end{itemize}
Here, ${\rm FP}$ is the class of functions 
from bit strings to integers that are computable in polynomial time by a Turing machine.
A ${\rm GapP}$ function~\cite{GapP}
is a function from bit strings to integers that is equal to
the number of accepting paths minus that of rejecting paths of a nondeterministic
Turing machine which takes the bit strings as input.
The FP function $f$ can be replaced with $2^{q(|w|)}$ for a polynomial 
$q$~\cite{GapP,Li},
and the error bound $(2^{-r(|w|)},1-2^{-r(|w|)})$ can be replaced with, 
for example, $(1/3,2/3)$~\cite{Fenner}.
\end{definition}

Interestingly, AWPP was shown to contain BQP,
by Fortnow and Rogers~\cite{FR} in 1997,
and since then it has been the best upperbound of BQP
(in classical complexity classes). 
Here, BQP is a class of problems efficiently solved by a quantum computer:

\begin{definition}
\label{def:BQP}
A language $L$ is in ${\rm BQP}$ iff there exists
a uniform family $V=\{V_n\}_n$ of polynomial-size quantum circuits 
such that
\begin{itemize}
\item[1.]
If $w\in L$ then $P_{V_w}(o=1)\ge \frac{2}{3}$.
\item[2.]
If $w\notin L$ then $P_{V_w}(o=1)\le \frac{1}{3}$.
\end{itemize}
Here,
we say that a family $V=\{V_n\}_n$ of quantum circuits is uniform
if there is a classical polynomial-time algorithm that outputs a description of $V_n$
on input $1^n$, where $n$ is the input size of $V_n$.
We denote the output bit by $o\in\{0,1\}$, and
$P_{V_w}(o=1)$ is the probability of obtaining 
$o=1$ (i.e., output 1)
if we measure the single output qubit of the circuit $V_{|w|}$ on
input $w$.
The pair of the thresholds $(\frac{1}{3},\frac{2}{3})$ is rather arbitrary.
For example, we can take
$(2^{-r(|w|)},1-2^{-r(|w|)})$ for any polynomial $r$.
\end{definition}
(We note that, for simplicity, we choose Hadamard and Toffoli gates 
as a universal gate set of quantum circuits.
This choice is crucial to obtain some of our results, while this choice is also taken
in Ref.~\cite{postBQP}, and we believe that this choice is enough to study
the essential parts of what we are interested in.
It may be possible to extend our results to other gate sets, but it would 
be a future research subject.)

The name of AWPP is thus known by many researchers including physicists.
However, the definition of AWPP seems to be somehow artificial and difficult to understand
for ones who are not familiar with GapP functions.
The purpose of the present contribution is to provide a
quantum physical interpretation of AWPP.
For the goal, we consider quantum computing with a postselection.
Here, a postselection is a (fictious) ability that we can choose
an event with probability 1 even if its probability is exponentially small.
Quantum computing with postselection was first considered by Aaronson~\cite{postBQP}.
He defined the following 
class postBQP, and showed that it is equal to PP 
(see also Ref.~\cite{Kuperberg} and Appendix~\ref{app4} for another proof of 
${\rm postBQP}={\rm PP}$):

\begin{definition}
\label{def:postBQP}
A language $L$ is in ${\rm postBQP}$ iff there exist a uniform family $V=\{V_n\}_n$
of polynomial-size quantum circuits with the ability of a postselection 
and a polynomial $u$ such that for any input $w$,
\begin{itemize}
\item[1.]
$P_{V_w}(p=1)\ge2^{-u(|w|)}$.
\item[2.]
If $w\in L$ then $P_{V_w}(o=1|p=1)\ge\frac{2}{3}$.
\item[3.]
If $w\notin L$ then $P_{V_w}(o=1|p=1)\le\frac{1}{3}$.
\end{itemize}
Here, $p\in\{0,1\}$ is the measurement result of the postselected qubit of the
circuit $V_{|w|}$, and $P_{V_w}(o=1|p=1)$ is the conditional probability that $V_{|w|}$
on input $w$ obtains $o=1$ under $p=1$.
Like ${\rm BQP}$, the pair of the thresholds $(\frac{1}{3},\frac{2}{3})$ is arbitrary. 
In particular, it can be $(2^{-r(|w|)},1-2^{-r(|w|)})$ for any polynomial $r$.
Furthermore, without loss of generality, we can assume that only a single
qubit is postselected, since postselections on more than two qubits
can be transformed to that on a single qubit by using the generalized Toffoli gate,
which can be implemented in a polynomial-size quantum circuit.
\end{definition}

We introduce a restricted version of postBQP, which we call 
${\rm postBQP}_{\rm aFP}$:
\begin{definition}
\label{def:postBQP_aFP}
A language $L$ is in ${\rm postBQP}_{\rm aFP}$ iff for any polynomials 
$r_1\ge0$ and $r_2\ge0$ there exist 
a uniform family $V=\{V_n\}_n$ of polynomial-size quantum circuits with the ability
of a postselection,
an ${\rm FP}$ function $f$,
and a polynomial $q$ such that for any input $w$,
$0<f(w)\le2^{q(|w|)}$
and
\begin{itemize}
\item[1.]
If $w\in L$ then $1-2^{-r_1(|w|)}\le P_{V_w}(o=1|p=1)\le 1$.
\item[2.]
If $w\notin L$ then $0\le P_{V_w}(o=1|p=1)\le 2^{-r_1(|w|)}$.
\item[3.]
$\Big|P_{V_w}(p=1)-\frac{f(w)}{2^{q(|w|)}}\Big|
\le 2^{-r_2(|w|)}P_{V_w}(p=1)$.
\end{itemize}
\end{definition}

The third condition intuitively means that the postselection probability $P_{V_w}(p=1)$
can be approximated to $f(w)/2^{q(|w|)}$
within the multiplicative error $2^{-r_2(|w|)}$.
(Hence the subscript ``aFP" means ``approximately FP".)
We show that ${\rm postBQP}_{\rm aFP}={\rm AWPP}$,
which provides a quantum physical interpretation of
AWPP: AWPP can be considered as an example of
postselected quantum complexity classes.
We note that while one might consider that ${\rm postBQP}_{\rm aFP}$ is also artificial
due to the fiction of postselection, we consider that this class is easier to
understand for physicists since it is defined by using the terminology of
quantum physics, or at least it gives another interpretation of AWPP,
which might be useful for future studies on AWPP.

We also introduce another restricted version of postBQP, which we call 
${\rm postBQP}_{\rm asize}$:
\begin{definition}
\label{def:postBQP_asize}
The definition of ${\rm postBQP}_{\rm asize}$ is the same as that
of ${\rm postBQP}_{\rm aFP}$
except that the ${\rm FP}$ function $f(w)$ is replaced with
$g(1^{|w|})$, where $g$ is a ${\rm GapP}$ function.
\end{definition}
We show that ${\rm postBQP}_{\rm asize}$ is equal to the classical
complexity class APP defined by Li~\cite{Li}.
Therefore, not only AWPP but also APP have quantum physical interpretations.

There are some researches
on quantum physical interpretations of classical complexity classes.
For example, the above mentioned
Aaronson's result ${\rm postBQP}={\rm PP}$~\cite{postBQP} is considered
as a quantum physical interpretation of PP.
Furthermore, Kuperberg~\cite{Kuperberg} showed that
${\rm A}_0{\rm PP}$ is equal to ${\rm SBQP}$, which is a quantum version
of SBP~\cite{BGM03},
and Fenner et al.~\cite{NQP} (see also Ref.~\cite{Yamakami}) showed that
${\rm coC}_={\rm P}$ is equal to
NQP, which is a quantum analogue of NP. 
Our contributions are in the same line of these researches,
while we take a different way for the proofs.
We not only use the relations between quantum computation and GapP functions
as used in Refs.~\cite{Fenner,NQP}, but combine them with the notion
of restricted postselection probability introduced in this paper.
Moreover, we also use tactically the property that AWPP and APP are closed under 
complement in  order to satisfy such a restriction of postselection probability.

In addition to 
${\rm postBQP}_{\rm aFP}$ and 
${\rm postBQP}_{\rm asize}$,
we introduce several restricted versions of postBQP, and study
relations among them and other complexity classes.
For example, we define a simpler version (the exact version) of ${\rm postBQP}_{\rm aFP}$,
which we call ${\rm postBQP}_{\rm FP}$:
\begin{definition}
\label{def:postBQP_FP}
A language $L$ is in ${\rm postBQP}_{\rm FP}$ iff it is in ${\rm postBQP}$
and there exist a polynomial $q$ and $f\in{\rm FP}$ $(f>0)$
such that for any input $w$,
$
P_{V_w}(p=1)=\frac{f(w)}{2^{q(|w|)}},
$
where $V$ is the uniform family of quantum circuits that assures $L\in{\rm postBQP}$.
\end{definition}
Since it is simpler than ${\rm postBQP}_{\rm aFP}$, it would be better if
we could show the equivalence of it to AWPP. Currently, we do not know
whether the equivalence holds. However,
we show that ${\rm postBQP}_{\rm FP}$ sits between WPP and AWPP.
(The definition of WPP is given in Sec.~\ref{pre}.)
It is nearly tight except showing the equivalence since WPP is one of the best
lower bounds of AWPP~\cite{toolkit} (in fact, AWPP was named as ``approximate WPP").
All our results are summarized in Fig.~\ref{suppfig}.
Definitions of new classes in the figure are given in Sec.~\ref{pre}.

A classical analogue of postBQP is ${\rm postBPP}$, which is known
to be equal to ${\rm BPP}_{\rm path}$~\cite{Han}.
We also consider a classical version, 
${\rm postBPP}_{\rm FP}$,
of
${\rm postBQP}_{\rm FP}$,
and show that 
${\rm UP}\cap{\rm coUP}\subseteq
{\rm postBPP}_{\rm FP}
\subseteq{\rm WAPP}$.
(The definitions of ${\rm postBPP}_{\rm FP}$ and WAPP are given in Sec.~\ref{pre}.)

\begin{figure}[htbp]
\begin{center}
\includegraphics[width=0.65\textwidth]{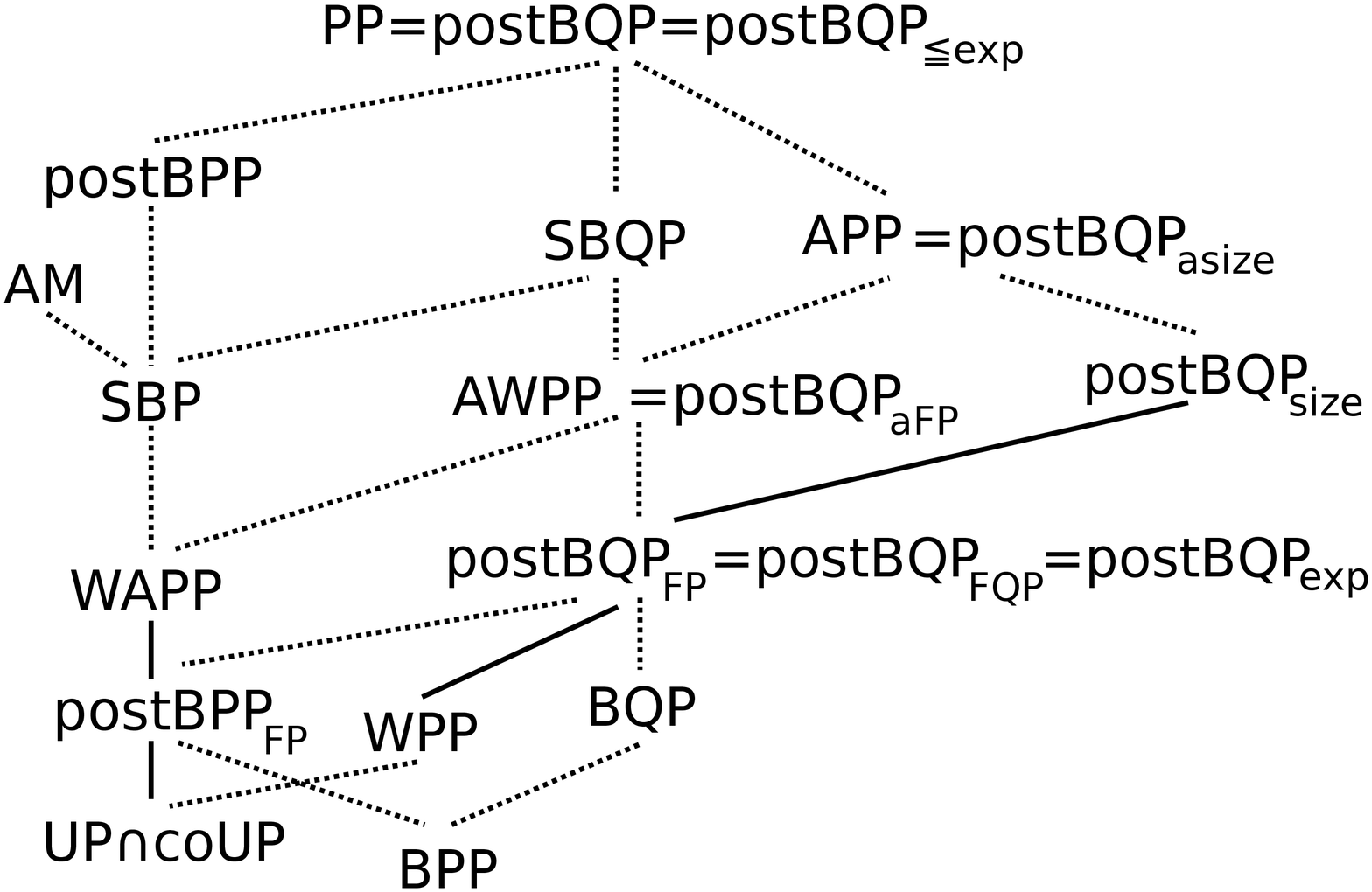}
\end{center}
\caption{
Relations among complexity classes studied in this paper.
Dotted lines are known results or inclusions followed
by definitions. Solid lines and all equalities (except
for ${\rm PP}={\rm postBQP}$) are
new results of this paper.}
\label{suppfig}
\end{figure}

\section{Preliminaries}
\label{pre}
In this section, we provide several definitions and facts
used in this paper.

\begin{definition}\cite{GapP}
\label{def:counting_machine}
A counting machine is a nondeterministic Turing machine running in polynomial
time with two halting states, accepting and rejecting, and every computation
path must end in one of these states.
Without loss of generality, we may assume each node of the computation tree
has outdegree at most two.
A counting machine is called normal if for any input each computational path has
the same number of nodes with outdegree two.
\end{definition}

\begin{definition} 
\label{def:sharpP}
A function $f:\{0,1\}^*\to {\mathbb N}\cup\{0\}$ is a ${\rm \#P}$ function
if there exists a counting 
machine $C$ such that
$f(w)$ is the number of accepting paths of $C(w)$,
where $C(w)$ denotes the
nondeterministic computation of $C$ on input $w$.
\end{definition}

\begin{definition}\cite{GapP} 
\label{def:GapP}
A function $f:\{0,1\}^*\to {\mathbb Z}$ is a ${\rm GapP}$ function
if there exists a counting machine $C$ 
such that
$f(w)$ is the number of accepting paths of $C(w)$
minus the number of rejecting paths of $C(w)$.
\end{definition}

\begin{definition}\cite{Li}
\label{def:APP}
A language $L$ is in ${\rm APP}$ iff for any polynomial $r$, there exist
$f,g\in{\rm GapP}$ such that for all $w$, $f(1^{|w|})>0$ and
\begin{itemize}
\item[1.]
If $w\in L$ then $1-2^{-r(|w|)}\le \frac{g(w)}{f(1^{|w|})}\le 1$.
\item[2.]
If $w\notin L$ then $0\le \frac{g(w)}{f(1^{|w|})}\le 2^{-r(|w|)}$.
\end{itemize}
\end{definition}

\begin{definition}\cite{BGM03}
\label{def:WAPP}
A language $L$ is in ${\rm WAPP}$ iff there exist $g\in{\rm \# P}$, a polynomial $p$,
and a constant $\epsilon>0$ such that
\begin{itemize}
\item[1.]
If $w\in L$ then 
$
\frac{1+\epsilon}{2}<\frac{g(w)}{2^{p(|w|)}}\le 1.
$
\item[2.]
If $w\notin L$ then 
$
0\le\frac{g(w)}{2^{p(|w|)}}< \frac{1-\epsilon}{2}.
$
\end{itemize}
Note that $2^{p(|w|)}$ can be replaced with an ${\rm FP}$ function $f(w)>0$.
\end{definition}

\begin{definition}\cite{GapP}
\label{def:WPP}
A language $L$ is in ${\rm WPP}$ iff there exist a ${\rm GapP}$ function $g$ 
and an ${\rm FP}$ function $f$ with $0\notin range(f)$ such that
\begin{itemize}
\item[1.]
If $w\in L$ then $g(w)=f(w)$.
\item[2.]
If $w\notin L$ then $g(w)=0$.
\end{itemize}
\end{definition}

There are relations between an output probability distribution
of a quantum circuit and a ${\rm GapP}$ function.
\begin{theorem}(Fortnow and Rogers~\cite{FR})
For any uniform family $V=\{V_n\}_n$ of polynomial-size quantum circuits,
there exist $g\in{\rm GapP}$ and a polynomial $q$ such that
for any $w$,
$
P_{V_w}(o=1)=\frac{g(w)}{2^{q(|w|)}}, 
$
where $P_{V_w}(o=1)$ is the probability that the output of
the circuit $V_{|w|}$ is $o=1$ on input $w$.
(Note that this theorem depends on the gate set.
As we have noted, in this paper, we consider the Hadamard and Toffoli gates
as a universal gate set.)
\label{PGapP}
\end{theorem}

\begin{theorem} (Fenner, Green, Homer, and Pruim~\cite{NQP})
For any $g\in{\rm GapP}$, there exist
a polynomial $s$ and
a uniform family 
$\{V_n\}_n$ 
of polynomial-size quantum circuits 
such that
$
P_{V_w}(o=1)=\frac{g(w)^2}{2^{s(|w|)}}.
$
\label{GapPP}
\end{theorem}

Now we introduce the restricted postBQP classes other than those
introduced in the previous section.
(Here, $V$ is the uniform family of polynomial-size quantum circuits
that assures $L\in{\rm postBQP}$ as in Definition~\ref{def:postBQP_FP}.)

\begin{definition}
\label{def:postBQP_size}
A language $L$ is in ${\rm postBQP}_{\rm size}$ iff it is in ${\rm postBQP}$ and
$P_{V_w}(p=1)$ depends only on $|w|$.
\end{definition}
From Theorem~\ref{PGapP}, it is an exact version of ${\rm postBQP}_{\rm asize}$.

\begin{definition} 
\label{def:postBQP_leexp}
A language $L$ is in ${\rm postBQP}_{\le{\rm exp}}$
iff it is in ${\rm postBQP}$ and there exists a polynomial $q>0$ such that
for any input $w$,
$
P_{V_w}(p=1)\le2^{-q(|w|)}.
$
\end{definition}

\begin{definition} 
\label{def:postBQP_exp}
A language $L$ is in ${\rm postBQP}_{\rm exp}$ iff
it is in ${\rm postBQP}$ and there exists a polynomial $q$ such
that for any input $w$,
$
P_{V_w}(p=1)=\frac{1}{2^{q(|w|)}}.
$
\end{definition}

\begin{definition} 
\label{def:postBQP_FQP}
A language $L$ is in ${\rm postBQP}_{\rm FQP}$ iff
it is in ${\rm postBQP}$ and there exist a polynomial $q$ and a function 
$f:\{0,1\}^*\to{\mathbb N}$, which can be calculated~\footnote{
We assume that $f$ can be calculated without any error.
} by a uniform family
of polynomial-size quantum circuits, such that for any input $w$, 
$
P_{V_w}(p=1)=\frac{f(w)}{2^{q(|w|)}}.
$
\end{definition}

We also consider the classical analogue of ${\rm postBQP}_{\rm FP}$.
\begin{definition}
\label{postBPP_FP}
We consider the following polynomial-time probabilistic Turing machine.
\begin{itemize}
\item[1.]
At every nondeterministic step, it makes a random decision between two possibilities, 
and each possibility is chosen with probability 1/2.
\item[2.]
The number of random decisions is the same for all computation paths.
\end{itemize}
Therefore, if the machine halts after $t$ nondeterministic steps, the probability of obtaining
a specific computation path is $2^{-t}$. 

A language $L$ is in ${\rm postBPP}_{\rm FP}$ iff there exist
a polynomial-time probabilistic Turing machine $V$ that satisfies the above properties
and outputs two bits $p$ and $o$,
an ${\rm FP}$ function $f>0$, a polynomial $q$, and a constant $\epsilon>0$ such that
\begin{itemize}
\item[1.]
$
P_{V_w}(p=1)=\frac{f(w)}{2^{q(|w|)}}.
$
\item[2.]
If $w\in L$ then 
$
\frac{1+\epsilon}{2}\le P_{V_w}(o=1|p=1)\le 1.
$
\item[3.]
If $w\notin L$ then 
$
0\le P_{V_w}(o=1|p=1)\le \frac{1-\epsilon}{2}.
$
\end{itemize}
Here, $P_{V_w}(p=1)$ and $P_{V_w}(o=1|p=1)$ are defined similarly to the case where
$V$ is a uniform family of circuits.
\end{definition}

\section{Results}

The main result of the present contribution is the following
quantum interpretation of AWPP:
\begin{theorem}
${\rm AWPP}={\rm postBQP}_{\rm aFP}$.
\label{main}
\end{theorem}
The proof is given in Sec.~\ref{proof1}.

By replacing some FP functions in the proof with GapP functions,
we can also show the following quantum interpretation of APP:
\begin{theorem}
\label{main2}
${\rm APP}={\rm postBQP}_{\rm asize}$.
\end{theorem}
The proof is given in Appendix~\ref{app1}.

If we consider not the approximate version, ${\rm postBQP}_{\rm aFP}$,
but the exact version, ${\rm postBQP}_{\rm FP}$, we do not know whether
it is equal to AWPP. Since 
${\rm postBQP}_{\rm FP}\subseteq{\rm postBQP}_{\rm aFP}$,
we know
${\rm postBQP}_{\rm FP}\subseteq{\rm AWPP}$.
Furthermore, we can show the following nearly tight lowerbound:
\begin{theorem}
${\rm WPP}\subseteq{\rm postBQP}_{\rm FP}$.
\label{WPPinpostBQP_FP}
\end{theorem}
The proof is given in Appendix~\ref{app2}.

We can also show several relations among
restricted postBQP classes:

\begin{theorem}
\label{postBQPequalpostBQP_leexp}
${\rm postBQP}={\rm postBQP}_{\le {\rm exp}}$.
\end{theorem}
The proof is given in Appendix~\ref{app3}.

\begin{theorem}
${\rm postBQP}_{\rm FP}={\rm postBQP}_{\rm FQP}={\rm postBQP}_{\rm exp}\subseteq
{\rm postBQP}_{\rm size}$.
\label{equivalences}
\end{theorem}
Its proof is given in Sec.~\ref{proof2}.

Finally, we consider the classical analogue, ${\rm postBPP}_{\rm FP}$,
of ${\rm postBQP}_{\rm FP}$, and show the following result:
\begin{theorem}
${\rm UP}\cap{\rm coUP}\subseteq{\rm postBPP}_{\rm FP}\subseteq{\rm WAPP}$.
\label{classical}
\end{theorem}
Its proof is given in Sec.~\ref{proof3}.
Note that the inclusion 
${\rm postBPP}_{\rm FP}\subseteq{\rm WAPP}$ is a ``classical analogue"
of
${\rm postBQP}_{\rm FP}\subseteq{\rm AWPP}$, since WAPP
is a ``\#P analogue" of AWPP.
Since ${\rm WAPP}\subseteq{\rm AM}$~\cite{BGM03} and ${\rm BQP}\subseteq{\rm AM}$
is unlikely, 
it is also unlikely that ${\rm BQP}\subseteq{\rm postBPP}_{\rm FP}$.
Furthermore, since it is unlikely that BQP contains
${\rm UP}\cap{\rm coUP}$,
the inclusion
${\rm UP}\cap{\rm coUP}\subseteq{\rm postBPP}_{\rm FP}$ suggests that
${\rm postBPP}_{\rm FP}={\rm BPP}$
and
${\rm postBPP}_{\rm FP}\subseteq{\rm BQP}$
are unlikely.

\section{Proof of Theorem \ref{main}}
\label{proof1}
We first show
${\rm AWPP}\cap{\rm coAWPP}\subseteq{\rm postBQP}_{\rm aFP}$.
Since ${\rm AWPP}={\rm coAWPP}$~\cite{Li}, this means
${\rm AWPP}\subseteq{\rm postBQP}_{\rm aFP}$.

Let us assume that a language $L$ is in 
${\rm AWPP}\cap{\rm coAWPP}$.
Then, for any polynomial $r$, there exist $g_1,g_2\in {\rm GapP}$
and $f_1,f_2\in{\rm FP}$ ($f_1>0$, $f_2>0$)
such that 
\begin{itemize}
\item[1.]
If $w\in L$ then
\begin{eqnarray*}
1-2^{-r(|w|)}\le \frac{g_1(w)}{f_1(w)}\le 1,~\mbox{and}~
0\le\frac{g_2(w)}{f_2(w)}\le 2^{-r(|w|)}.
\end{eqnarray*}
\item[2.]
If $w\notin L$ then
\begin{eqnarray*}
0\le \frac{g_1(w)}{f_1(w)}\le 2^{-r(|w|)},~\mbox{and}~
1-2^{-r(|w|)}\le \frac{g_2(w)}{f_2(w)}\le 1.
\end{eqnarray*}
\end{itemize}
In the following, 
for simplicity, we omit the $|w|$ dependency of $r$,
and just write $r(|w|)$ as $r$.

Then, there exist two GapP functions 
$h_1(w)\equiv g_1(w)f_2(w)$
and
$h_2(w)\equiv g_2(w)f_1(w)$, such that
\begin{itemize}
\item[1.]
If $w\in L$ then
\begin{eqnarray*}
1-2^{-r}\le \frac{h_1(w)}{f_1(w)f_2(w)}\le 1,~\mbox{and}~
0\le\frac{h_2(w)}{f_1(w)f_2(w)}\le 2^{-r}.
\end{eqnarray*}
\item[2.]
If $w\notin L$ then
\begin{eqnarray*}
0\le \frac{h_1(w)}{f_1(w)f_2(w)}\le 2^{-r},~\mbox{and}~
1-2^{-r}\le \frac{h_2(w)}{f_1(w)f_2(w)}\le 1.
\end{eqnarray*}
\end{itemize}
Then there exist two counting machines $C^1$ and $C^2$
such that
$h_1(w)=C_a^1(w)-C_r^1(w)$ and
$h_2(w)=C_a^2(w)-C_r^2(w)$,
where $C_a^j(w)$ and $C_r^j(w)$ $(j=1,2)$ are the numbers of 
accepting and rejecting paths of $C^j$ on input $w$,
respectively.

There exist two normal counting machines $N^1$ and $N^2$ such that
$h_1(w)=\frac{1}{2}(N_a^1(w)-N_r^1(w))$ and
$h_2(w)=\frac{1}{2}(N_a^2(w)-N_r^2(w))$~\cite{GapP}.
Without loss of generality, we can assume that computation paths of $N^1$ and $N^2$
on input $w$ can be represented by strings in $\{0,1\}^{q(|w|)}$,
where $q$ is a polynomial.
(In the following, for simplicity, we write $q(|w|)$ as $q$.)
Then we consider a uniform family $V=\{V_n\}_n$ of quantum circuits defined by the following
procedure on input $w$.
First, the state
\begin{eqnarray*}
\frac{
|0\rangle^{\otimes 2k}
}{\sqrt{2^{q+1}}}\sum_{x\in\{0,1\}^{q}}|x\rangle
\Big((-1)^{N^1(w,x)}|N^1(w,x)\rangle|1\rangle
+(-1)^{N^2(w,x)}|N^2(w,x)\rangle|0\rangle\Big)
\end{eqnarray*}
can be generated by a polynomial-size quantum circuit.
Here, $k$ is a polynomial chosen later ($k$ precisely means $k(|w|)$),
and $N^j(w,x)=0$ (=1, resp.) if the path $x$ of $N^j$ on input $w$
is an accepting (rejecting, resp.) one.
Let us postselect the first, second, and third registers to 
$|+\rangle^{\otimes 2k+q+1}$.
The (unnormalized) state on the last register, which is the output qubit,
after the postselection is
\begin{eqnarray*}
\frac{1}{2^{q+1+k}}
\Big((N_a^1(w)-N_r^1(w))|1\rangle
+(N_a^2(w)-N_r^2(w))|0\rangle
\Big),
\end{eqnarray*}
and therefore 
\begin{eqnarray*}
P_{V_w}(p=1)=\frac{(N_a^1(w)-N_r^1(w))^2+(N_a^2(w)-N_r^2(w))^2}
{2^{2q+2+2k}}
								=\frac{4(h_1^2(w)+h_2^2(w))}
{2^{2q+2+2k}}.
\end{eqnarray*}
Therefore, irrespective of $w\in L$ or $w\notin L$,
we obtain
\begin{eqnarray*}
\frac{f_1^2(w)f_2^2(w)}{2^{2q+2k}}(1-2^{-r})^2
\le
P_{V_w}(p=1)\le\frac{f_1^2(w)f_2^2(w)}{2^{2q+2k}}(1+2^{-2r}).
\end{eqnarray*}
Let us define $s(w)=f_1^2(w)f_2^2(w)$. Then the above inequality
means
\begin{eqnarray*}
\frac{s(w)}{2^{2q+2k}}(1-2^{-r})^2
\le
P_{V_w}(p=1)\le\frac{s(w)}{2^{2q+2k}}(1+2^{-2r}).
\end{eqnarray*}
Since
$
1-2^{-r+1}
\le
(1-2^{-r})^2
$
and
$
1+2^{-2r}
\le
1+2^{-r+1},
$
we obtain
\begin{eqnarray*}
\frac{s(w)}{2^{2q+2k}}(1-2^{-r+1})
\le
P_{V_w}(p=1)
\le\frac{s(w)}{2^{2q+2k}}(1+2^{-r+1}),
\end{eqnarray*}
which means, if we take $r\ge2$,
\begin{eqnarray}
\frac{P_{V_w}(p=1)}{1+2^{-r+1}}
\le
\frac{s(w)}{2^{2q+2k}}
\le
\frac{P_{V_w}(p=1)}
{1-2^{-r+1}}.
\label{equation1}
\end{eqnarray}
Note that
\begin{eqnarray}
\frac{1}{1-2^{-r+1}}
&\le&1+2^{-r+2},
\label{equation2}
\end{eqnarray}
and
\begin{eqnarray}
\frac{1}{1+2^{-r+1}}
-(1-2^{-r+2})
=\frac{1}{1+2^{-r+1}}
(2^{-r+1}+2^{-2r+3})
\ge0.
\label{equation3}
\end{eqnarray}
Therefore, from Eqs.~(\ref{equation2}) and (\ref{equation3}), 
Eq.~(\ref{equation1}) becomes
\begin{eqnarray*}
(1-2^{-r+2})P_{V_w}(p=1)\le\frac{s(w)}{2^{2q+2k}}\le
(1+2^{-r+2})P_{V_w}(p=1),
\end{eqnarray*}
which means
\begin{eqnarray*}
\Big|P_{V_w}(p=1)
-\frac{s(w)}{2^{2q+2k}}\Big|
\le 2^{-r+2}P_{V_w}(p=1).
\end{eqnarray*}
Remember that $s(w)= f_1^2(w)f_2^2(w)>0$ and it is in FP.
We denote $t\equiv 2q+2k$ and take $k$ such that
$s(w)\le 2^{t}$.
For any polynomial $r_2$, let us take $r\ge r_2+2$.
Then,
\begin{eqnarray*}
\Big|P_{V_w}(p=1)
-\frac{s(w)}{2^{t}}\Big|
\le 2^{-r_2}P_{V_w}(p=1).
\end{eqnarray*}

Furthermore, from the state after the postselection, we have
\begin{eqnarray*}
P_{V_w}(o=1|p=1)=\frac
{(N_a^1(w)-N_r^1(w))^2}
{(N_a^1(w)-N_r^1(w))^2
+(N_a^2(w)-N_r^2(w))^2
}
=\frac{h_1^2(w)}{h_1^2(w)+h_2^2(w)}.
\end{eqnarray*}
For any polynomial $r_1$, let us take $r\ge r_1+2$.
Then, if $w\in L$ we obtain
\begin{eqnarray*}
P_{V_w}(o=1|p=1)=\frac{h_1^2(w)}{h_1^2(w)+h_2^2(w)}
												\ge\frac{(1-2^{-r})^2}{1+2^{-2r}}
												\ge1-2^{-r_1},
\end{eqnarray*}
and if $w\notin L$ we obtain
\begin{eqnarray*}
P_{V_w}(o=1|p=1)=\frac{h_1^2(w)}{h_1^2(w)+h_2^2(w)}
												\le\frac{2^{-2r}}{(1-2^{-r})^2}
												\le2^{-r_1}.
\end{eqnarray*}
Therefore, by taking $r\ge \max(r_1+2,r_2+2)$, $L$ is in ${\rm postBQP}_{\rm aFP}$.

Next we show
${\rm postBQP}_{\rm aFP}\subseteq{\rm AWPP}$.
Let us assume that a language $L$ is in ${\rm postBQP}_{\rm aFP}$.
Then for any polynomials $r_1$ and $r_2$
there exist a uniform family $V=\{V_n\}_n$ of polynomial-size quantum circuits,
an FP function $f$, and a polynomial $q$ satisfying the condition in 
Definition~\ref{def:postBQP_aFP}.
From Theorem~\ref{PGapP}, there exist a GapP function $g$
and a polynomial $s$ such that 
$
P_{V_w}(o=1,p=1)=\frac{g(w)}{2^{s}},
$
where $P_{V_w}(o=1,p=1)$ is the joint probability distribution for $o$ and $p$.
Therefore, if we take $r_2\ge1$, we obtain
\begin{itemize}
\item[1.]
If $w\in L$ then
$
P_{V_w}(p=1)(1-2^{-r_1})\le P_{V_w}(o=1,p=1)\le P_{V_w}(p=1),
$
which means
\begin{eqnarray*}
\frac{f(w)}{2^{q}(1+2^{-r_2})}(1-2^{-r_1})
\le \frac{g(w)}{2^{s}}\le 
\frac{f(w)}{2^{q}(1-2^{-r_2})},
\end{eqnarray*}
and therefore
\begin{eqnarray*}
\frac{1-2^{-r_2}}{1+2^{-r_2}}(1-2^{-r_1})
\le \frac{g(w)2^{q}(1-2^{-r_2})}{2^{s}f(w)}\le 1.
\end{eqnarray*}

\item[2.]
If $w\notin L$ then
$
0\le P_{V_w}(o=1,p=1)\le 2^{-r_1}P_{V_w}(p=1),
$
which means
\begin{eqnarray*}
0\le \frac{g(w)}{2^{s}}
\le 2^{-r_1}\frac{f(w)}{2^{q}(1-2^{-r_2})},
\end{eqnarray*}
and therefore
\begin{eqnarray*}
0\le \frac{g(w)2^{q}(1-2^{-r_2})}{2^{s}f(w)}
\le 2^{-r_1}.
\end{eqnarray*}
\end{itemize}
Note that
\begin{eqnarray*}
\frac{g(w)2^{q}(1-2^{-r_2})}{2^{s}f(w)}
=
\frac{g(w)2^{q}(2^{r_2}-1)}{2^{s+r_2}f(w)},
\end{eqnarray*}
and we can see
$g(w)2^{q}(2^{r_2}-1)\in {\rm GapP}$,
$2^{s+r_2}f(w)>0$,
and $2^{s+r_2}f(w)\in {\rm FP}$.

If we take $r_1=r_2\ge3$,
$
\frac{(1-2^{-r_1})^2}{1+2^{-r_1}}
\ge
\frac{2}{3}
$,
and $2^{-r_1}\le \frac{1}{3}$.
Therefore $L$ is in AWPP due to the definition of
AWPP in Ref.~\cite{Fenner}.

\section{Proof of Theorem \ref{equivalences}}
\label{proof2}

The inclusions 
${\rm postBQP}_{\rm exp}\subseteq{\rm postBQP}_{\rm size}$ 
and
${\rm postBQP}_{\rm FQP}\supseteq{\rm postBQP}_{\rm FP}\supseteq{\rm postBQP}_{\rm exp}$
are obvious.
Let us show
${\rm postBQP}_{\rm FQP}\subseteq{\rm postBQP}_{\rm exp}$.
Its proof uses the idea of an additive adjustment of the acceptance probability from
Ref.~\cite{JKNN12}
with a standard multiplicative adjustment.

Let us assume that a language $L$ is in ${\rm postBQP}_{\rm FQP}$.
Then, there exist a uniform family $V=\{V_n\}_n$ of polynomial-size quantum circuits,
a function $f:\{0,1\}^*\to {\mathbb N}$ whose $f(w)$ can be calculated
by another uniform family of polynomial-size quantum circuits for any input $w$, and a polynomial $h\ge0$ such that
$
P_{V_w}(p=1)=\frac{f(w)}{2^{h}}
$
($h$ precisely means $h(|w|)$)
and
\begin{itemize}
\item[1.]
If $w\in L$, then
$\frac{9}{10}\le P_{V_w}(o=1|p=1)\le1$.
\item[2.]
If $w\notin L$, then
$0\le P_{V_w}(o=1|p=1)\le \frac{1}{10}$.
\end{itemize}

We can take a function $t:\{0,1\}^*\to{\mathbb N}\cup\{0\}$ 
such that
$
2^{t(w)}\le f(w)<2^{t(w)+1}
$
for any input $w$. Note that
$t(w)$ can be calculated by a uniform family of polynomial-size quantum circuits.

From $V$, we construct the uniform family $W=\{W_n\}_n$ of
polynomial-size quantum circuits implemented on input $w$ as follows:
\begin{itemize}
\item[1.]
$W_{|w|}$ flips a coin. If heads, it simulates $V_{|w|}$.
\item[2.]
If tails, $W_{|w|}$ outputs $o=1$ with probability 1/2, and
$p=1$ with probability 
$
\frac{2^{t(w)+1}-f(w)}{2^{h}}
$
independently.
\end{itemize}
Since
\begin{eqnarray*}
2^{h}-2^{t(w)+1}+f(w)\ge f(w)-2^{t(w)+1}+f(w)
																										=2(f(w)-2^{t(w)})
																										\ge0,
\end{eqnarray*}
we obtain
$
\frac{2^{t(w)+1}-f(w)}{2^{h}}\le1.
$

Then, 
\begin{eqnarray*}
P_{W_w}(p=1)=\frac{1}{2}P_{V_w}(p=1)+\frac{1}{2}\frac{2^{t(w)+1}-f(w)}{2^{h}}
												=\frac{2^{t(w)}}{2^{h}},
\end{eqnarray*}
and
\begin{eqnarray*}
P_{W_w}(o=1|p=1)&=&\frac{P_{W_w}(o=1,p=1)}{P_{W_w}(p=1)}\\
																&=&\frac{\frac{1}{2}P_{V_w}(o=1|p=1)P_{V_w}(p=1)
+\frac{1}{2}\frac{2^{t(w)+1}-f(w)}{2^{h}}\frac{1}{2}}
{\frac{2^{t(w)}}{2^{h}}}\\
&=&\frac{f(w)}{2^{t(w)+1}}P_{V_w}(o=1|p=1)
+\frac{1}{2}\frac{2^{t(w)+1}-f(w)}{2^{t(w)+1}}.
\end{eqnarray*}

If $w\in L$,
\begin{eqnarray*}
P_{W_w}(o=1|p=1)
&=&\frac{f(w)}{2^{t(w)+1}}P_{V_w}(o=1|p=1)
+\frac{1}{2}\frac{2^{t(w)+1}-f(w)}{2^{t(w)+1}}\\
&\ge&\frac{1}{2}\frac{9}{10}
+\frac{1}{2}\frac{1}{2}
=\frac{7}{10}.
\end{eqnarray*}
If $w\notin L$,
\begin{eqnarray*}
P_{W_w}(o=1|p=1)
&=&\frac{f(w)}{2^{t(w)+1}}P_{V_w}(o=1|p=1)
+\frac{1}{2}\frac{2^{t(w)+1}-f(w)}{2^{t(w)+1}}\\
&\le&\frac{1}{2}\frac{1}{10}
+\frac{1}{2}\frac{1}{2}
=\frac{3}{10}.
\end{eqnarray*}
Here, we have used the fact that
$
\alpha\frac{9}{10}+(1-\alpha)\frac{1}{2}\ge
\frac{1}{2}\frac{9}{10}+\frac{1}{2}\frac{1}{2}
$
and
$
\alpha\frac{1}{10}+(1-\alpha)\frac{1}{2}\le
\frac{1}{2}\frac{1}{10}+\frac{1}{2}\frac{1}{2}
$
for $\alpha\ge 1/2$.
Note that $f(w)/2^{t(w)+1}\ge1/2$, since $f(w)\ge2^{t(w)}$.

From $W$, we construct the uniform family 
$R=\{R_n\}_n$ of polynomial-size quantum circuits implemented on input $w$
in the following way:
\begin{itemize}
\item[1.]
$R_{|w|}$ simulates $W_{|w|}$.
\item[2.]
$R_{|w|}$ outputs $o=1$ if and only if $W_{|w|}$ outputs $o=1$.
\item[3.]
$R_{|w|}$ generates a random bit $b$ which takes $b=1$ with 
probability $2^{-t(w)}$.
(Note that $t(w)\le h$.)
\item[4.]
$R_{|w|}$ outputs $p=1$ if and only if $b=1$ and $W_{|w|}$ outputs $p=1$.
\end{itemize}
Then,
$
P_{R_w}(o=1|p=1)=P_{W_w}(o=1|p=1)
$
and
$
P_{R_w}(p=1)=P_{W_w}(p=1)2^{-t(w)}=2^{-h}.
$
Therefore, $L$ is in ${\rm postBQP}_{\rm exp}$.

\section{Proof of Theorem \ref{classical}}
\label{proof3}
Let us first show 
${\rm postBPP}_{\rm FP}\subseteq{\rm WAPP}$.
We assume that a language $L$ is in ${\rm postBPP}_{\rm FP}$. 
Then, there exist a probabilistic Turing machine $V$,
an FP function $f>0$, and a polynomial $s$ such that
$
P_{V_w}(p=1)=\frac{f(w)}{2^{s}}.
$
There exist a \#P function $g$ and a polynomial $q$ such that
$
P_{V_w}(o=1,p=1)=\frac{g(w)}{2^{q}}.
$
Therefore, by the conditions on $P_{V_w}(o=1|p=1)$,
we obtain if $w\in L$,
$
\frac{1+\epsilon}{2}\le \frac{2^{s}g(w)}{2^{q}f(w)}\le1,
$
and if $w\notin L$,
$
0\le \frac{2^{s}g(w)}{2^{q}f(w)}\le \frac{1-\epsilon}{2}.
$
Since $2^{s}g(w)$ is a \#P function and $2^{q}f(w)$ is an FP function, $L$ is in WAPP.

Now let us show
${\rm UP}\cap{\rm coUP}\subseteq{\rm postBPP}_{\rm FP}$.
Let us assume that a language $L$ is in 
${\rm UP}\cap{\rm coUP}$.
Then, there exist two polynomial-time nondeterministic Turing machines $N$ and $M$ such that
\begin{itemize}
\item[1.]
If $w\in L$ then 
$N$ has exactly one accepting path,
and all paths of $M$ reject. 
\item[2.]
If $w\notin L$ then 
all paths of $N$ reject, and
$M$ has exactly one accepting path.
\end{itemize}
Without loss of generality, we can assume that both $N$ and $M$
have $2^{q}$ computation paths.
Let us consider the following algorithm $V$:
\begin{itemize}
\item[1.]
Randomly choose $x\in\{0,1\}^{q}$, and simulate
the computation paths represented by $x$ of $N$ and $M$ on input $w$.
\item[2.]
If both $N$ and $M$ reject, output $p=0$ and $o=0$.
If $N$ accepts and $M$ rejects, output $p=1$ and $o=1$.
If $M$ accepts and $N$ rejects, output $p=1$ and $o=0$.
\item[3.]
Postselect on $p=1$.
\end{itemize}
The probability of postselecting $p=1$ is $2^{-q}$.
Furthermore,
$P_{V_w}(o=1|p=1)=
1$ if $w\in L$,
and it is 0
if $w\notin L$.
Therefore, $L$ is in ${\rm postBPP}_{\rm FP}$.

{\bf Acknowledgements}.
TM is supported by the Tenure Track System by MEXT Japan,
the JSPS Grant-in-Aid for Young Scientists (B) No.26730003, and 
the MEXT JSPS Grant-in-Aid for Scientific Research on Innovative Areas No.15H00850.
HN is supported by the JSPS Grant-in-Aid for Scientific Research (A) 
Nos.23246071, 24240001, 26247016,
and (C) No.25330012,
and the MEXT JSPS Grant-in-Aid for Scientific Research on
Innovative Areas No.24106009. 
We acknowledge an anonymous reviewer for pointing out a possibility
of improving the lowerbound of ${\rm postBQP}_{\rm FP}$
in an early draft of this paper.

\appendix
\section{Proof of Theorem~\ref{main2}}
\label{app1}


The proof is the same as that of
${\rm postBQP}_{\rm aFP}={\rm AWPP}$ (Theorem~\ref{main})
given in Sec.~\ref{proof1}.

First, we show
${\rm APP}\subseteq{\rm postBQP}_{\rm asize}$.
Since ${\rm APP}={\rm coAPP}$~\cite{Li}, 
we show
${\rm APP}\cap{\rm coAPP}\subseteq{\rm postBQP}_{\rm asize}$.
The rest of the proof
is the same as that of
${\rm AWPP}\cap{\rm coAWPP}\subseteq{\rm postBQP}_{\rm aFP}$
except that two FP functions $f_1(w)$ and $f_2(w)$ are replaced with
two GapP functions $f_1(1^{|w|})$ and $f_2(1^{|w|})$.

Furthermore, the proof of
${\rm postBQP}_{\rm asize}\subseteq{\rm APP}$
is also the same as that of
${\rm postBQP}_{\rm aFP}\subseteq{\rm AWPP}$.
We have only to replace the FP function $f(w)$
with a GapP function $f(1^{|w|})$.

\section{Proof of Theorem~\ref{WPPinpostBQP_FP}}
\label{app2}

Since ${\rm WPP}={\rm coWPP}$, we show
${\rm WPP}\cap{\rm coWPP}\subseteq{\rm postBQP}_{\rm FP}$.
Let us assume that a language $L$ is in 
${\rm WPP}\cap{\rm coWPP}$. Then, there exist GapP functions $g_1$ and $g_2$,
and FP functions $f_1$ and $f_2$ with $0\notin range(f_1)$ and $0\notin range(f_2)$ 
such that
\begin{itemize}
\item[1.]
If $w\in L$
\begin{eqnarray*}
g_1(w)&=&f_1(w),\\
g_2(w)&=&0.
\end{eqnarray*}
\item[2.]
If $w\notin L$
\begin{eqnarray*}
g_1(w)&=&0,\\
g_2(w)&=&f_2(w).
\end{eqnarray*}
\end{itemize}

Then, there exist GapP functions 
$g_1'(w)\equiv g_1(w)f_2(w)$ and $g_2'(w)\equiv g_2(w)f_1(w)$ such that
\begin{itemize}
\item[1.]
If $w\in L$
\begin{eqnarray*}
g_1'(w)&=&g_1(w)f_2(w)=f_1(w)f_2(w),\\
g_2'(w)&=&g_2(w)f_1(w)=0.
\end{eqnarray*}
\item[2.]
If $w\notin L$
\begin{eqnarray*}
g_1'(w)&=&g_1(w)f_2(w)=0,\\
g_2'(w)&=&g_2(w)f_1(w)=f_2(w)f_1(w).
\end{eqnarray*}
\end{itemize}
In other words, there exist counting machines $C^1$ and $C^2$
such that
\begin{itemize}
\item[1.]
If $w\in L$
\begin{eqnarray*}
C_a^1(w)-C_r^1(w)&=&f_1(w)f_2(w),\\
C_a^2(w)-C_r^2(w)&=&0.
\end{eqnarray*}
\item[2.]
If $w\notin L$
\begin{eqnarray*}
C_a^1(w)-C_r^1(w)&=&0,\\
C_a^2(w)-C_r^2(w)&=&f_2(w)f_1(w).
\end{eqnarray*}
\end{itemize}
Here, $C_a^j(w)$ and $C_r^j(w)$ are numbers of accepting and rejecting
paths of $C^j$ on input $w$, respectively.

There exist normal counting machines $N^1$ and $N^2$ such that~\cite{GapP}
\begin{itemize}
\item[1.]
If $w\in L$
\begin{eqnarray*}
N_a^1(w)-N_r^1(w)&=&2f_1(w)f_2(w),\\
N_a^2(w)-N_r^2(w)&=&0.
\end{eqnarray*}
\item[2.]
If $w\notin L$
\begin{eqnarray*}
N_a^1(w)-N_r^1(w)&=&0,\\
N_a^2(w)-N_r^2(w)&=&2f_1(w)f_2(w).
\end{eqnarray*}
\end{itemize}
Without loss of generality, we can assume that both $N^1$ and $N^2$ have
computation trees on input $w$ whose paths are represented by $\{0,1\}^{q(|w|)}$.

For a given input $w$, $V=\{V_n\}_n$ is defined as the following procedure.
First, we generate 
\begin{eqnarray*}
\frac{1}{\sqrt{2^{q(|w|)+1}}}\sum_{x\in\{0,1\}^{q(|w|)}}
\Big(&&(-1)^{N^1(w,x)}|x\rangle\otimes|N^1(w,x)\rangle\otimes|1\rangle\\
&&+(-1)^{N^2(w,x)}|x\rangle\otimes|N^2(w,x)\rangle\otimes|0\rangle
\Big)
\end{eqnarray*}
by a polynomial-size quantum circuit.
Let us postselect the first and second registers on $|+\rangle^{\otimes q(|w|)+1}$.
Then, the (unnormalized) state after the postselection is
\begin{eqnarray*}
\frac{1}{2^{q(|w|)+1}}
\Big((N^1_a(w)-N^1_r(w))|1\rangle
+(N^2_a(w)-N^2_r(w))|0\rangle
\Big).
\end{eqnarray*}
Therefore, the postselection probability is
\begin{eqnarray*}
P_{V_w}(p=1)&=&\frac{(N_a^1(w)-N_r^1(w))^2+(N_a^2(w)-N_r^2(w))^2}{2^{2q(|w|)+2}}\\
								&=&\frac{(2f_1(w)f_2(w))^2}{2^{2q(|w|)+2}}\ge \frac{1}{2^{2q(|w|)}}.
\end{eqnarray*}
Furthermore,
\begin{eqnarray*}
P_{V_w}(o=1|p=1)=
\left\{
\begin{array}{ll}
1&(w\in L),\\
0&(w\notin L).
\end{array}
\right.
\end{eqnarray*}
Therefore, $L$ is in ${\rm postBQP}_{\rm FP}$.

\section{Proof of Theorem~\ref{postBQPequalpostBQP_leexp}}
\label{app3}

${\rm postBQP}\supseteq{\rm postBQP}_{\le{\rm exp}}$ is obvious.
Let us show
${\rm postBQP}\subseteq{\rm postBQP}_{\le {\rm exp}}$.
We assume that a language $L$ is in ${\rm postBQP}$.
Then, from the uniform family $V=\{V_n\}_n$ of polynomial-size quantum circuits
that assures $L\in{\rm postBQP}$, we construct the
uniform family $W=\{W_n\}_n$ of polynomial-size quantum circuits
which run as follows on input $w$: $W_{|w|}$ generates a random bit $b$
which is $b=1$ with probability $2^{-q(|w|)}$, where $q>0$ is any polynomial.
Then, $W_{|w|}$ simulates $V_{|w|}$ and outputs $p=1$ 
if $b=1$ and $V_{|w|}$ outputs $p=1$.
$W_{|w|}$ outputs $o=1$ 
if $V_{|w|}$ outputs $o=1$.

Then,
\begin{eqnarray*}
P_{W_w}(p=1)=P_{V_w}(p=1)2^{-q(|w|)}\le 2^{-q(|w|)}
\end{eqnarray*}
and
\begin{eqnarray*}
P_{W_w}(o=1|p=1)=P_{V_w}(o=1|p=1).
\end{eqnarray*}
Therefore, $L$ is in ${\rm postBQP}_{\le {\rm exp}}$.

\section{Another proof of ${\rm postBQP}={\rm PP}$}
\label{app4}
Here we give another proof of
${\rm postBQP}={\rm PP}$.
Before showing the proof, we will give  
two definitions of PP.

A standard definition of PP is as follows.

\begin{definition}
A language $L$ is in PP iff there exists a polynomial-time non-deterministic
Turing machine such that
\begin{itemize}
\item[1.]
If $w\in L$ then at least 1/2 of computation paths accept.
\item[2.]
If $w\notin L$ then less than 1/2 of computation paths accept.
\end{itemize}
\end{definition}

There is another definition of PP that we will use:

\begin{definition} (Fortnow~\cite[Theorem 6.4.16]{Li})
A language $L$ is in PP
iff
for any polynomial $r$, there exist $f,g\in{\rm GapP}$ such that
$f>0$ and
\begin{itemize}
\item[1.]
If $w\in L$ then $1-2^{-r(|w|)}\le\frac{g(w)}{f(w)}\le 1$.
\item[2.]
If $w\notin L$ then $0\le\frac{g(w)}{f(w)}\le 2^{-r(|w|)}$.
\end{itemize}
\end{definition}

\begin{theorem} (Aaronson~\cite{postBQP})
${\rm PP}={\rm postBQP}$.
\end{theorem}

\begin{proof}
First we show ${\rm postBQP}\subseteq{\rm PP}$.
We assume that a language $L$ is in postBQP. Then,
for any polynomial $r$, there exists a uniform family $\{V_n\}_n$
of polynomial-size quantum circuits.
As in the proof of ${\rm postBQP}_{\rm FP}\subseteq{\rm AWPP}$,
if $w\in L$, 
\begin{eqnarray*}
&&1-2^{-r}\le P_{V_w}(o=1|p=1)\le 1\\
&\Leftrightarrow&
1-2^{-r}\le\frac{P_{V_w}(o=1,p=1)}{P_{V_w}(p=1)}\le 1\\
&\Leftrightarrow&
1-2^{-r}\le\frac{g(w)2^{q'}}{2^{q}f(w)}\le 1
\end{eqnarray*}
for $f,g\in {\rm GapP}$ and polynomials $q$ and $q'$.
Here, we have used the fact from Theorem~\ref{PGapP} that 
\begin{eqnarray*}
P_{V_w}(o=1,p=1)&=&\frac{g(w)}{2^{q}}\\
P_{V_w}(p=1)&=&\frac{f(w)}{2^{q'}}
\end{eqnarray*}
for some $g,f\in{\rm GapP}$ and polynomials $q$ and $q'$.
Note that for simplicity, we omit the $|w|$ dependencies of some polynomials.

If $w\notin L$
\begin{eqnarray*}
&&0\le P_{V_w}(o=1|p=1)\le 2^{-r}\\
&\Leftrightarrow&
0\le\frac{P_{V_w}(o=1,p=1)}{P_{V_w}(p=1)}\le 2^{-r}\\
&\Leftrightarrow&
0\le\frac{g(w)2^{q'}}{2^{q}f(w)}\le 2^{-r}.
\end{eqnarray*}
Since $2^{q'}g(w),2^{q}f(w)\in{\rm GapP}$,
$L$ is in PP.

Second, let us show ${\rm PP}\subseteq{\rm postBQP}$.
We assume that a language $L$ is in PP.
If $w\in L$, for any polynomial $r$, there exist $g,f\in{\rm GapP}$ such that
\begin{eqnarray*}
&&(1-2^{-r})^2\le\frac{g(w)^2}{f(w)^2}.
\end{eqnarray*}
Then, from Theorem~\ref{GapPP}, we have
\begin{eqnarray*}
	P_{W_w'}(o=1)=2^{-q'}f(w)^2,\\
	P_{V_w'}(o=1)=2^{-q}g(w)^2,
\end{eqnarray*}
which means
\begin{eqnarray*}
(1-2^{-r})^2\le\frac{2^{q}P_{V'_w}(o=1)}{2^{q'}P_{W'_w}(o=1)}
\end{eqnarray*}
for some polynomials $q$ and $q'$, and
uniform families $\{V_n'\}_n$ and $\{W_n'\}_n$ of polynomial-size quantum circuits.
Let us define $V_{|w|}$ and $W_{|w|}$
such that
\begin{eqnarray*}
P_{V_w}(o=1)&=&P_{V_w'}(o=1)2^{-q'},\\
P_{W_w}(o=1)&=&P_{W_w'}(o=1)2^{-q}.
\end{eqnarray*}
The circuit $V_{|w|}$ ($W_{|w|}$) can be constructed by simulating $V_{|w|}'$
($W_{|w|}'$) and outputting $o=1$ with probability $2^{-q'}$ 
($2^{-q}$) if and only if $V_{|w|}'$ ($W_{|w|}'$) outputs $o=1$. 
Then, we obtain
\begin{eqnarray*}
(1-2^{-r})^2\le\frac{P_{V_w}(o=1)}{P_{W_w}(o=1)}.
\end{eqnarray*}

Similarly, if $w\notin L$,
we have
\begin{eqnarray*}
&&\frac{g(w)^2}{f(w)^2}\le 2^{-2r}\\
&\Leftrightarrow&
\frac{P_{V_w}(o=1)}{P_{W_w}(o=1)}\le 2^{-2r}.
\end{eqnarray*}

Let us consider the following quantum circuit $R_n$:
It first flips two unbiased coins. If both are heads, $R_n$ simulates $W_n$. 
\begin{itemize}
\item[1.]
If $W_n$ outputs $o=1$, then $R_n$ outputs $o=0$ and $p=1$.
\item[2.]
If $W_n$ outputs $o=0$, then $R_n$ outputs $o=0$ and $p=0$. 
\end{itemize}
Otherwise, $R_n$ simulates $V_n$.
\begin{itemize}
\item[1.]
If $V_n$ outputs $o=1$, then $R_n$ outputs $o=1$ and $p=1$.
\item[2.]
If $V_n$ outputs $o=0$, then $R_n$ outputs $o=0$ and $p=0$. 
\end{itemize}

Then, 
\begin{eqnarray*}
P_{R_w}(p=1)&=&\frac{3}{4}P_{V_w}(o=1)+\frac{1}{4}P_{W_w}(o=1)\\
												&\ge&\frac{f(w)^2}{4\times2^{q+q'}}\\
												&>&\frac{1}{2^{q+q'+2}},
\end{eqnarray*}
and
\begin{eqnarray*}
P_{R_w}(o=1|p=1)&=&\frac{P_{R_w}(o=1,p=1)}{P_{R_w}(p=1)}\\
																&=&\frac{\frac{3}{4}P_{V_w}(o=1)}{\frac{3}{4}P_{V_w}(o=1)+\frac{1}{4}P_{W_w}(o=1)}.
\end{eqnarray*}

If $w\in L$,
\begin{eqnarray*}
P_{R_w}(o=1|p=1)&=&\frac{3P_{V_w}(o=1)}{3P_{V_w}(o=1)+P_{W_w}(o=1)}\\
&\ge&\frac{3P_{V_w}(o=1)}{3P_{V_w}(o=1)+\frac{P_{V_w}(o=1)}{(1-2^{-r})^2}}\\
&=&\frac{3-6\times2^{-r}+3\times2^{-2r}}
	{4-6\times2^{-r}+3\times2^{-2r}}\\
&\ge&\frac{3-6\times2^{-r}}
	{4+3\times\frac{1}{2}}\\
&\ge&\frac{1}{2}+\frac{1}{22}-\frac{12}{11}\times 2^{-r}.
\end{eqnarray*}
If $w\notin L$,
\begin{eqnarray*}
P_{R_w}(o=1|p=1)&=&\frac{3P_{V_w}(o=1)}{3P_{V_w}(o=1)+P_{W_w}(o=1)}\\
																&\le&\frac{3P_{V_w}(o=1)}{3P_{V_w}(o=1)+\frac{P_{V_w}(o=1)}{2^{-2r}}}\\
																&\le&3\times 2^{-2r}.
\end{eqnarray*}
Therefore, $L\in {\rm postBQP}$.
\end{proof}

\end{document}